\documentclass{article}

\usepackage{amssymb}
\usepackage{amsmath}
\usepackage{amsthm}
\usepackage{graphics}       
\usepackage{graphicx}       
\usepackage[space]{grffile} 
\usepackage{subfig}			
\usepackage{color} 			
\usepackage{url}
\usepackage{hyperref}
 \usepackage{fullpage}

\newtheorem{definition}{Definition}
\newtheorem{theorem}{Theorem}
\newtheorem{proposition}{Proposition}

\theoremstyle{remark}

\newcommand{\la}{\ensuremath{ \mathcal{L}}}
\newcommand{\PLB}{\ensuremath{ \mathcal{P}_{\alpha}}}
\newcommand{\PLC}{\ensuremath{ \mathcal{P}'_{\alpha}}}
\newcommand{\Sparse}{\ensuremath{ \mathcal{S}_{c,n}}}

\date{February 12, 2015}

\title{Near-optimal adjacency labeling scheme for power-law graphs}
\author{Casper Petersen, Noy Rotbart,\\ Jakob Grue Simonsen and Christian Wulff-Nilsen \\ \\
\small{Department of Computer Science, University of Copenhagen} \\
\small{Universitetsparken 5, 2100 Copenhagen}\\
 \small{\texttt{\{cazz,noyro,simonsen,koolooz\}@diku.dk}} \\
 }

\begin{document}
\begin{titlepage}
\clearpage\maketitle
\thispagestyle{empty}
\begin{abstract}
An adjacency labeling scheme is a method that assigns labels to the vertices of a graph such that adjacency between vertices 
can be inferred directly from the assigned label, without using a  centralized data structure.
We devise adjacency labeling schemes for  the family  of power-law graphs. This  family that  has been used to model many types of networks, e.g. the Internet AS-level graph. Furthermore, we  prove an almost matching  lower bound for this family.
We also provide an asymptotically near-optimal  labeling scheme for sparse graphs.
Finally, we validate the efficiency of our labeling scheme by  an experimental evaluation  using both synthetic data and real-world networks of up to hundreds of thousands of vertices. 

\end{abstract}
\end{titlepage}
\newpage
\section{Introduction}
A fundamental problem in networks is how to disseminate the structural information of the underlying graph of a network to its vertices. The purpose of such dissemination is that the local topology of the network can be inferred using only local information stored in each vertex without using costly access to large, global data structures.
One way of doing so is via  \emph{ labeling schemes}: an algorithm that assigns a bit string--a \emph{label}--to each vertex so that a query between any two vertices can be deduced solely from their respective labels. 
The main objective of  labeling schemes is to minimize the \emph{maximum label size}: the maximum number of bits used in a label of any vertex. 
Labeling schemes for adjacency and other properties have found practical use in  XML search engines~\cite{cohen2010labeling}, mapping services~\cite{abraham2011hub} and routing~\cite{krioukov2004compact}.

In this paper we are interested in particular with labeling schemes for   \emph{adjacency} queries. 
For general graphs Moon~\cite{moon1965minimal} showed lower and upper bounds of respectively $n/2$ and $n/2+\log n$ bits on the  label size.
The asymptotic gap between these bounds was only recently closed  by Alstrup et al.~\cite{alstrup2014adjacency} who proved an upper bound of $n/2 + 6$ bits. 
Upper bounds for adjacency labeling schemes exist for many specific classes of graphs, including trees~\cite{Alstrup02}, planar graphs~\cite{gavoille2007shorter},  bounded-degree graphs~\cite{adjiashvili2014labeling}, and bipartite graphs~\cite{lozin2007minimal}.

However, for classes of graphs whose statistical properties--in particular their \emph{degree distribution}--more closely resemble that of real-world networks, there has, to our knowledge, been no research on adjacency labeling schemes.
One class of graphs extensively used for modelling real-world networks is \emph{power-law graphs}: roughly, $n$-vertex graphs where the number of vertices of degree $k$ is proportional to $n/k^{\alpha}$ for some positive $\alpha$. Power-law graphs (also called scale-free graphs in the literature) have been used, e.g., to model the Internet AS-level graph \cite{DBLP:journals/ton/SiganosFFF03,DBLP:conf/podc/AkellaCKS03}, and many other types of network (see, e.g., \cite{mitzenmacher2004brief,clauset2009power} for overviews). 
The adequacy of fit of power-law graph models to actual data, as well as the empirical correctness of the conjectured mechanisms giving rise to power-law behaviour, have been subject to criticism (see, e.g., \cite{DBLP:journals/jacm/AchlioptasCKM09,clauset2009power}). 
In spite of such criticism, and because their degree distribution affords a reasonable approximation of the degree distribution of many networks, the class of power-law graphs remains a popular tool in network modelling whose statistical behaviour is well-understood: e.g., for power-law graphs with $2<\alpha<3$, the range most often seen in the modeling of real-world networks \cite{clauset2009power}, it is known that with high probability the average distance between any two vertices is  $O(\log \log n)$, the diameter is $O(\log n)$ and there exists a dense subgraph of $n^{c/\log \log n}$ vertices~\cite{chung2004average}. 

Routing labeling schemes for power-law graphs  have been investigated by Brady and Cowen~\cite{brady2006compact}, and by Chen et al.~\cite{chen2012compact}. Labeling schemes for other properties than adjacency have been investigated for various classes of graphs, e.g., distance~\cite{gavoillea2004distance}, and flow~\cite{katz2004labeling}. 
Dynamic labeling schemes were studied by Korman and Peleg~\cite{korman2007compact,Korman07,korman2007general} and recently by Dahlgaard et. al~\cite{dahlgaard2014dynamic}.
Experimental evaluation for some labeling schemes for various properties on general graphs have been performed by Caminiti et.~al~\cite{caminiti2008engineering}, Fischer~\cite{fischer2009short} and Rotbart et.~al~\cite{rotbart2014evaluation}.

Adjacency labeling schemes are tightly coupled with  the graph-theory related concept of induced universal graphs.
Given a  graph family  $\mathcal{F}$, the aim  is to find smallest $N$ such that a graph of  $N$ vertices contains all graphs in $\mathcal{F}$ as induced subgraphs. 
Kannan, Naor and Rudich~\cite{Kannan92} showed that an $f(n) \log n$ adjacency labeling scheme for $\mathcal{F}$  constructs an induced universal graph for this family of  $2^{f(n)}$ vertices. Some  of the adjacency labeling schemes reported earlier  contributed a better bound than was known of  induced universal graphs (see e.g~\cite{BCLR,Alstrup02}).
In the context of  sparse graphs,  a body of work on universal graphs\footnote{A graph that  contains each graph from the graph family as subgraph, not necessarily induced.} for this family was  investigated both by  Babai et al.~\cite{babai1982graphs} and  by Alon and Asodi~\cite{Alon2002universal}.

\subsection{Our contribution}

Our contributions are:

\paragraph{An  $O(\sqrt[\alpha] n (\log n)^{1 - 1/\alpha})$ adjacency labeling scheme for power-law graphs $G$.}
The scheme is based on two ideas:
(I) a labeling \emph{strategy} that  partitions the vertices of $G$ into high (``fat'') and low degree (``thin'') vertices based on a threshold degree, and (II) a threshold \emph{prediction} that depends only on the coefficient $\alpha$ of a power-law curve fitted to the degree distribution of $G$. 
Real-world power-law graphs rarely exceed  $~10^{10}$ vertices, implying a label size of at most  ${10^{5}}$ bits, well within the processing capabilities of current hardware. 
We claim that our  scheme is thus appealing in practice   due both to  its simplicity and hte small size of its labels.
Using the same ideas, we get an  asymptotically near-tight  $O(\sqrt{n \log n})$ adjacency labeling scheme for sparse graphs.

\paragraph{A lower bound of $\Omega(\sqrt[\alpha]{n})$ bits on the maximum label size for \emph{any} adjacency labeling scheme for power-law graphs.}
To this end we define a  restrictive subclass of power-law graphs and show that it is contained in the bigger class we study for the upper bound; we show that this class requires label size $\Omega(\sqrt[\alpha]{n})$ for $n$-vertex graphs.
This lower bound shows that our upper bound above is asymptotically  optimal, bar a $(\log n)^{1 - 1/\alpha}$ factor.
By the connections between adjacency labeling schemes and universal graphs, we also obtain upper and lower bounds for induced universal graphs for power-law graphs.

\paragraph{An experimental investigation  of our labeling scheme}
Using both real-world (23K-325K vertices) and synthetic (300K-1M vertices) data sets, we observe that:
(i) Our threshold \emph{prediction} performs close to optimal when using the labeling \emph{strategy} above. 
(ii) our labeling scheme achieves maximum label size several orders of magnitude smaller than the state-of-the-art labeling schemes for more general graph families.
\vspace{\baselineskip}

In addition, our study  may contribute to  the understanding of the quality of  \emph{generative models}---procedures that ``grow'' random graphs whose degree distributions are with high probability ``close'' to power-law graphs,  such as the Barabasi-Albert model~\cite{barabasi1999emergence} and the   Aiello-Chung-Lu model~\cite{aiello2001random}. As a first step, we provide an evidence  that the randomized Barabasi-Albert model~\cite{barabasi1999emergence} produces only a small fraction of the power-law graphs possible.
 \section{Graph Families Related to Power-Law Graphs}\label{Sec:GraphFamilies}
In this section we define two families of graphs $\PLB$ and $\PLC$ with $\PLC\subseteq\PLB$. Family $\PLB$ is rich enough to contain the graphs whose degree distribution is approximately, or perfectly, power-law distributed, and our upper bound on the label size for our labeling scheme holds for any graph in $\PLB$. Family $\PLC$ is used to show our lower bound. In the following, let $i_1 = \Theta(\sqrt[\alpha]n)$ be the smallest integer such that $\lfloor Cn/i_1^\alpha\rfloor \leq 1$, and let $C'\geq(\frac C{\alpha-1} + \frac{i_1}{\sqrt[\alpha] n} + 5)^{\alpha} + \frac{C}{\alpha - 1}$ be a constant; we shall use $C'$ in the remainder of the paper.
\begin{definition} \label{def:general-family}
Let $\alpha > 1$ be a real number. $\PLB$ is the family of graphs $G$ such that if $n = \vert V(G)\vert$ then for all integers $k$ between $\sqrt[\alpha]{n/\log n}$ and $n-1$, $\sum_{i = k}^{n-1} {\vert V_i\vert} \leq C'(\frac{n}{k^{\alpha-1}})$.
\end{definition}

The class of $\alpha$\emph{-proper} power law graphs contains graphs where  the number of vertices of degree $k$
must be $C \frac{n}{k^{\alpha}}$ rounded either up or down and the number of vertices of degree $k$ is non-increasing
with $k$. Note that the function $k \mapsto  C \frac{1}{k^{\alpha}}$ is strictly decreasing.
 
\begin{definition}\label{def:proper}
Let $\alpha > 1$ be a real number. We say that an $n$-vertex graph  $G=(V,E)$ is an  $\alpha$-\emph{proper power-law graph} 
if
\begin{enumerate}
\item $\lfloor Cn\rfloor - i_1 - 1\leq\vert V_1\vert\leq\lceil Cn\rceil$,
\item $\lfloor C\frac n{2^\alpha}\rfloor\leq\vert V_2\vert\leq\lceil C\frac n{2^\alpha}\rceil + 1$,
\item for every $i$ with $3 \leq i \leq n$:  $\vert V_i \vert\in \{\lfloor C\frac{n}{i^{\alpha}} \rfloor, \lceil C \frac{n}{i^{\alpha}} \rceil\}$, and
\item for every $i$ with $2 \leq i \leq n-1$: $\vert V_i \vert \geq \vert V_{i+1} \vert$.
\end{enumerate}
The family of $\alpha$-proper power-law graphs is denoted $\PLC$.
\end{definition}
Note that we allow slightly more noise in the sizes of $V_1$ and $V_2$ than in the remaining sets; without it, it seems tricky to prove a better lower bound than $\Omega(\sqrt[\alpha+1]{n})$ on label sizes.

We show the following properties of $\PLC$. 
\begin{proposition}\label{prop:maxvertexproper}\label{prop:maxdegree}
The maximum degree in an $n$-vertex graph in $\PLC$ is at most $\left(\frac{C}{\alpha - 1} + 2\right) \sqrt[\alpha]{n} + i_1 + 3 = \Theta(\sqrt[\alpha]n)$.
\end{proposition}

\begin{proof}
Let $n > 0$ be an integer and let $k' = \lfloor \sqrt[\alpha]{n} \rfloor$. 
Furthermore, let $S_{k'} = \sum_{i=1}^{k'} \vert V_i\vert$, that is $S_{k'}$ is the number of vertices of degree at most $k'$. Let $S^{-}_{k'} = (\sum_{i=1}^{k'} \lfloor Cni^{-\alpha}\rfloor) - i_1 - 1$. Then
$S_{k'} \geq S^{-}_{k'}$. We now bound $S^{-}_{k'}$ from below.
For every $i$ with $1 \leq i \leq k'$,
\begin{align*}
S^{-}_{k'} + k' & = -i_1 - 1 + \sum_{i=1}^{k'} \left(\left\lfloor Cn i^{-\alpha}\right\rfloor + 1\right) \geq  -i_1 - 1 + \sum_{i=1}^{k'} Cn i^{- \alpha}  = -i_1 - 1 + Cn \sum_{i=1}^{k'} i^{-\alpha} \\
& \geq n \left(1 - C\sum_{i=k'+1}^{\infty} i^{-\alpha} \right) - i_1 - 1 
 \geq n \left( 1 - C\int_{k'}^\infty x^{-\alpha} dx \right) - i_1 - 1\\
 & = n \left( 1 - C\left[ \frac{1}{\alpha - 1} x^{-\alpha + 1}\right]_{\infty}^{k'}\right) - i_1 - 1 = n \left( 1 - \frac{C}{\alpha - 1} \left( \lceil\sqrt[\alpha]{n} \rceil \right)^{-\alpha + 1}\right) - i_1 - 1\\
 & \geq n \left( 1 -  \frac{C}{\alpha - 1} \left(\sqrt[\alpha]{n}\right)^{-\alpha + 1} \right) - i_1 - 1 = n - \frac{Cn}{\alpha - 1}n^{-1+\frac{1}{\alpha}} - i_1 - 1\\
& = n - \frac{C}{\alpha - 1}\sqrt[\alpha]{n} - i_1 - 1,
\end{align*}
giving $S_{k'} \geq S^{-}_{k'} \geq n - \frac{C}{\alpha - 1}\sqrt[\alpha]{n} - \lceil \sqrt[\alpha]{n} \rceil - i_1 - 1$. There are thus at most $\frac{C}{\alpha - 1} \sqrt[\alpha]{n} + \lfloor \sqrt[\alpha]{n} \rfloor + i_1 + 1$ vertices of degree strictly more than $k' = \lceil \sqrt[\alpha]{n} \rceil$. Since for every $1 \leq i \leq n-1$: $\vert V_i \vert \geq \vert V_{i+1} \vert$, it follows that the maximum degree of any $\alpha$-proper power-law graph is at most $\left(\frac{C}{\alpha - 1} + 2\right) \sqrt[\alpha]{n} + i_1 + 3$.
\end{proof}

\begin{proposition}\label{prop:powerlawsparse}
For $\alpha > 2$, all graphs in $\PLC$ are sparse.
\end{proposition}
\begin{proof}
By Proposition \ref{prop:maxvertexproper}, the maximum degree of an $n$-vertex $\alpha$-proper power-law
graph is at most $k' \triangleq \left(\frac{C}{\alpha - 1} + 2\right) \sqrt[\alpha]{n} + i_1 + 3$, whence
the total number of edges is at most $\frac{1}{2}\sum_{k=1}^{k'} k \vert V_k\vert$. By definition,
$\vert V_k \vert\leq \lceil \frac{Cn}{k^\alpha}\rceil \leq \frac{Cn}{k^{\alpha}} + 1$ for $k\neq 2$ and $\vert V_2\vert\leq\lceil\frac{Cn}{2^{\alpha}}\rceil + 1$, and thus
\begin{align*}
\frac{1}{2}\sum_{k=1}^{k'} k\vert V_k\vert &\leq 1 + \frac{1}{2}\sum_{k=1}^{k'} k \left(\frac{Cn}{k^{\alpha}} + 1 \right)
 \leq 1 + \frac{k'(k'+1)}{4} + Cn\sum_{k=1}^{\infty} k^{-\alpha+1} \\
&   = O(n^{2/\alpha}) + Cn \zeta(\alpha - 1) = O(n).
\end{align*}
\end{proof}

\begin{proposition}\label{prop:Contained}
$\PLC \subseteq \PLB$.
\end{proposition}
\begin{proof}
Let $d = \lfloor(\frac C{\alpha - 1} + 2)\sqrt[\alpha]{n} + i_1 + 3\rfloor$. For any $\alpha$-proper power-law graph with $n$ vertices and for any $k$, $\vert V_k\vert\leq Ck^{-\alpha}n + 1$ and by Proposition~\ref{prop:maxvertexproper}, $\vert V_k\vert = 0$ when $k > d$.

Let $k$ be an arbitrary integer between $\sqrt[\alpha]{n/\log n}$ and $n-1$. We need to show that $\sum_{i = k}^{n-1} {\vert V_i\vert} \leq C'(\frac{n}{k^{\alpha-1}})$. It suffices to show this for $k\leq d$. We have:
\begin{align*}
  \sum_{i = k}^{n-1} {\vert V_i\vert} & \leq \sum_{i = k}^d(Ci^{-\alpha}n + 1)
    =    d - k + 1 + Cn\sum_{i = k}^d i^{-\alpha}\\
  & \leq \left(\frac C{\alpha - 1} + \frac{i_1}{\sqrt[\alpha]n} + 5\right)\sqrt[\alpha]{n} + Cn\int_k^d x^{-\alpha}dx\\
  & \leq \left(\frac C{\alpha - 1} + \frac{i_1}{\sqrt[\alpha]n} + 5\right)\sqrt[\alpha]{n} + Cn\left[\frac 1{\alpha - 1}x^{-\alpha+1}\right]_\infty^k\\
  & \leq \left(\left(\frac C{\alpha - 1} + \frac{i_1}{\sqrt[\alpha]n} + 5\right)\left(\frac{\sqrt[\alpha]nd^{\alpha-1}}n\right) + \frac {C}{\alpha - 1}\right)nk^{-\alpha +1}\\
  & \leq \left(\left(\frac C{\alpha - 1} + \frac{i_1}{\sqrt[\alpha]n} + 5\right)\left(\frac C{\alpha-1} + \frac{i_1}{\sqrt[\alpha]n} + 5\right)^{\alpha - 1} + \frac{C}{\alpha - 1}\right)nk^{-\alpha+1} \leq C'nk^{-\alpha+1},
\end{align*}
as desired.
\end{proof}

\section{The Labeling Schemes}
\label{sec:lab_schem}
We now construct algorithms for labeling schemes for $c$-sparse graphs and for the family $\PLB$. Both labeling schemes partition vertices into \emph{thin} vertices which are of low degree and \emph{fat} vertices of high degree. The \emph{degree threshold} for the scheme is the lowest possible degree of a fat vertex. We start with $c$-sparse graphs.
\begin{theorem}\label{sparse-label}
There is a $\sqrt{2cn\log n} + 2\log n + 1$ labeling scheme for $\Sparse$.
\end{theorem}
\begin{proof}
Let $G=(V,E)$ be an $n$-vertex $c$-sparse graph. Let $f(n)$ be the degree threshold for $n$-vertex graphs; we choose $f(n)$ below. Let $k$ denote the number of fat vertices of $G$, and assign each to each fat vertex a unique identifier between $1$ and $k$. Each thin vertex is given a unique identifier between $k+1$ and $n$.

For a $v\in V$, the first part of the label $\la(v)$ is a single bit indicating whether $v$ is thin or fat followed by a string of $\log n$ bits representing its identifier. If $v$ is thin, the last part of $\la(v)$ is the concatenation of the identifiers of the neighbors of $v$. If $v$ is fat, the last part of $\la(v)$ is a \emph{fat bit string} of length $k$ where the $i$th bit is $1$ iff $v$ is incident to the (fat) vertex with identifier $i$.

Decoding a pair $(\la(u),\la(v))$ is now straightforward: if one of the vertices, say $u$, is thin, $u$ and $v$ are adjacent iff the identifier of $v$ is part of the label of $u$. If both $u$ and $v$ are fat then they are adjacent iff the $i$th bit of the fat bit string of $\la(u)$ is $1$ where $i$ is the identifier of $v$.

Since $|E|\leq cn$, we have $k\leq 2cn/f(n)$. A fat vertex thus has label size $1 + \log n + k\leq 1 + \log n + 2cn/f(n)$ and a thin vertex has label size at most $1 + \log n + f(n)\log n$. To minimize the maximum possible label size, we solve $2cn/x = x\log n$. Solving this gives $x = \sqrt{2cn/\log n}$ and setting $f(n) = \lceil x\rceil$ gives a label size of at most $1 + \log n + (\sqrt{2cn/\log n} + 1)\log n\leq 1 + 2\log n + \sqrt{2cn\log n}$.
\end{proof}


By Proposition~\ref{prop:powerlawsparse}, graphs in $\PLC$ are sparse for $\alpha > 2$. This gives a label size of $O(\sqrt{n\log n})$ with the labeling scheme in Theorem~\ref{sparse-label}. We now show that this label can be significantly improved, by constructing a labeling scheme for $\PLB$ which contains $\PLC$.

\begin{theorem}\label{prop:labelingMain}
 There is a $\sqrt[\alpha]{C'n}(\log n)^{1 - 1/\alpha} + 2\log n + 1$ labeling scheme for $\PLB$.
\end{theorem}
\begin{proof}
The proof is very similar to that of Theorem~\ref{sparse-label}. We let $f(n)$ denote the degree threshold. If we pick $f(n)\geq \sqrt[\alpha]{n/\log n}$ then by Definition~\ref{def:general-family}  there are at most $C'n / f(n)^{\alpha -1}$ fat vertices. Defining labels in the same way as in Theorem~\ref{sparse-label} gives a label size for thin vertices of at most $1 + \log n + f(n)\log n$ and a label size for fat vertices of at most
$1 + \log n + C'n / f(n)^{\alpha -1}$.
We minimize by solving
$x \log n = C'n / x^{\alpha -1}$, giving $x = \sqrt[\alpha]{C'n/\log n}$. Setting $f(n) = \lceil x\rceil$ gives a label size of at most $\sqrt[\alpha]{C'n}(\log n)^{1 - 1/\alpha} + 2\log n + 1$.
\end{proof}

\section{Lower Bounds}
We now derive lower bounds for the label size of any  labeling schemes for both $\Sparse$ and  $\PLC$.
Our proofs rely on  Moon's~\cite{moon1965minimal} lower bound of  $\lfloor n/2 \rfloor$ bits for labeling scheme for general graphs.
We first show that the upper bound achieved for sparse graphs is close to the best possible.
The following proposition is essentially a more precise version of the lower bound suggested by Spinrad~\cite{spinrad2003efficient}.
\begin{proposition}
Any  labeling scheme for $\Sparse$ requires  labels of size at least $\left\lfloor\frac{\sqrt{cn}}2\right\rfloor$ bits.
\end{proposition}
\begin{proof}
Assume for contradiction that there exists a labeling scheme  assigning labels of size strictly less than $\lfloor\frac{\sqrt{cn}}2\rfloor$.
Let $G$ be an $n$-vertex graph. Let $G'$ be the graph resulting by adding $\left\lfloor\frac{n^2}{c}\right\rfloor - n$ isolated vertices to $G$, and note that now $G'$ is $c$-sparse. The graph $G$ is an induced subgraph of  $G'$.
It now follows that the vertices of $G$ have  labels of size strictly less than 
$\left\lfloor\frac{\sqrt{c\lfloor n^2/c\rfloor}}2\right\rfloor \leq n/2$ bits. As $G$ was arbitrary, we obtain a contradiction.
\end{proof}

\subsection{Lower bound for power-law graphs}
In the remainder of this section we are assuming that $\alpha>2$ and  prove the following:
\begin{theorem}\label{centralLowerBound}
For all $n$, any labeling scheme for $n$-vertex graphs of $\PLB$ requires label size $\Omega(\sqrt[\alpha]{n})$.
\end{theorem}
More precisely, we present a lower bound for $\PLC$ which is contained in $\PLB$. Let $n\in\mathbb N$ be given and let $H = (V(H),E(H))$ be an arbitrary graph with $i_1$ vertices where $i_1 = \Theta(\sqrt[\alpha]n)$ is defined as in Section~\ref{Sec:GraphFamilies}. We show how to construct an $\alpha$-proper power-law graph $G = (V,E)$ with $n$ vertices that contains $H$ as an induced subgraph. Observe that a labeling of $G$ induces a labeling of $H$. As $H$ was chosen arbitrarily and as any labeling scheme for $k$-vertex graphs requires $\lfloor i_1/2 \rfloor$ label size in the worst case, Theorem~\ref{centralLowerBound} follows if we can show the existence of $G$.

We construct $G$ incrementally where initially $E = \emptyset$. Partition $V$ into subsets $V_1,\ldots,V_n$ as follows. The set $V_1$ has size $\lfloor Cn\rfloor - i_1$. For $i = 2,\ldots,i_1-1$, $V_i$ has size $\lfloor Cn/i^\alpha\rfloor$. Letting $n' = \sum_{i = 1}^{i_1-1} \vert V_i\vert$, we set the size of $V_i$ to $1$ for $i = i_1,\ldots,i_1+n-n'-1$ and the size of $V_i$ to $0$ for $i = i_1+n-n',\ldots,n$, thereby ensuring that the sum of sizes of all sets is $n$. Observe that $\sum_{i = 1}^{i_1}\lfloor Cn/i^\alpha\rfloor\leq n$ so that $n'\leq n - i_1$, implying that $n-n'\geq i_1$. Hence we have at least $i_1$ size $1$ subsets $V_{i_1},\ldots,V_{i_1+n-n'-1}$ in each of which the vertex degree allowed by Definition~\ref{def:proper} is at least $i_1$.

Let $v_1,\ldots,v_{i_1}$ be an ordering of $V(H)$, form a set $V_H\subseteq V$ of $i_1$ arbitrary vertices from the sets $V_{i_1},\ldots,V_{i_1+n-n'-1}$, and choose an ordering $v_1',\ldots,v_{i_1}'$ of $V_H$. For all $i,j\in\{1,\ldots,i_1\}$, add edge $(v_i',v_j')$ to $E$ iff $(v_i,v_j)\in E(H)$. Now, $H$ is an induced subgraph of $G$ and since the maximum degree of $H$ is $i_1-1$, no vertex of $V_i$ exceeds the degree bound allowed by Definition~\ref{def:proper} for $i = 1,\ldots,n$.

We next add additional edges to $G$ in three phases to ensure that it is an $\alpha$-proper power law graph while maintaining the property that $H$ is an induced subgraph of $G$. For $i = 1,\ldots,n$, during the construction of $G$ we say that a vertex $v\in V_i$ is \emph{unprocessed} if its degree in the current graph $G$ is strictly less than $i$. If the degree of $v$ is exactly $i$, $v$ is \emph{processed}.

\paragraph{Phase $1$:}
Let $V' = V\setminus (V_1\cup V_H)$. Phase $1$ is as follows: while there exists a pair of unprocessed vertices $(u,v)\in V'\times V_H$, add $(u,v)$ to $E$.

When Phase $1$ terminates, $H$ is clearly still an induced subgraph of $G$. Furthermore, all vertices of $V_H$ are processed. To see this, note that the sum of degrees of vertices of $V_H$ when they are all processed is $O(i_1^2) = O(n^{2/\alpha})$ which is $o(n)$ since $\alpha > 2$. Furthermore, prior to Phase $1$, each of the $\Theta(n)$ vertices of $V'$ have degree $0$ and can thus have their degrees increased by at least $1$ before being processed.

\paragraph{Phase $2$:}
Phase $2$ is as follows: while there exists a pair of unprocessed vertices $(u,v)\in V'\times V'$, add $(u,v)$ to $E$. At termination, at most one vertex of $V'$ remains unprocessed. If such a vertex exists we process it by connecting it to $O(\sqrt[\alpha]n)$ vertices of $V_1$; as $\vert V_1\vert = \Theta(n)$ there are enough vertices of $V_1$ to accomodate this. Furthermore, prior to adding these edges, all vertices of $V_1$ have degree $0$, and hence the bound allowed for vertices of this set is not exceeded.

\paragraph{Phase $3$:}
In Phase $3$, we add edges between pairs of unprocessed vertices of $V_1$ until no such pair exists. If no unprocessed vertices remain we have the desired $\alpha$-proper power law graph $G$. Otherwise, let $w\in V_1$ be the unprocessed vertex of degree $0$. We add a single edge from $w$ to another vertex $w'$ of $V_1$, thereby processing $w$ and moving $w'$ from $V_1$ to $V_2$. Note that the sizes of $V_1$ and $V_2$ are kept in their allowed ranges due to the first two conditions in Definition~\ref{def:proper}. This proves Theorem~\ref{centralLowerBound}.
\section{Scale Free  Graphs from Generative Models}\label{Sec:ScaleFree}
The Barab{\'a}si-Albert (BA) model is a well-known generative model for power-law graphs that, roughly, grows a graph in a sequence of time steps by
inserting a single vertex at each step and attaching it to $m$ existing vertices with probability weighted by the degree of each existing vertex \cite{barabasi1999emergence}. The BA model
generates graphs that asymptotically have a power-law degree distribution ($\alpha = 3$) for low-degree nodes \cite{DBLP:journals/rsa/BollobasRST01}.
Graphs created by the BA model have low arboricity (the arboricity of a graph is the minimum number of spanning forests needed to cover its edges.)~\cite{goel2006bounded}; we use
that fact to prove the following highly efficient labeling scheme. 


\begin{proposition}\label{Th:baLabeling}
The family of graphs generated by the BA model has an $O(m \log n)$ adjacency labeling scheme.
\end{proposition}

\begin{proof}
Let  $G=(V,E)$  be an $n$-vertex graph resulting by the construction  by the BA model with some parameter $m$ (starting from some graph $G_0 = (V_0,E_0)$ with $\vert V_0 \vert \ll n$).
While it is not known how to compute the   arboricity of a graph efficiently, it is possible in near-linear time to compute a partition of $G$ with  at most twice\footnote{More precisely, for any $\epsilon \in (0,1)$  there exist an $O(|E(G)| / \epsilon)$ algorithm~\cite{kowalik2006approximation} that computes such partition using at most $(1+ \epsilon)$ times more forests than the optimal.} the number of forests in comparison to the optimal~\cite{arikati1997efficient}.
We can thus decompose the graph to $2m$ forests in near linear time and label each forest using Alstrup and Rauhe's~\cite{Alstrup02} $\log n + O(\log^* n)$ labeling scheme for trees,  and achieve a $2m (\log n + O(\log^* n))$ labeling scheme for $G$.
\end{proof}

Note that if the encoder operates at the same time as the creation of the graph, Proposition \ref{Th:baLabeling} can be strengthened to yield an an $m \log n$ labeling scheme: simply store the  
 identifiers of the $m$ vertices attached with every vertex insertion.
Theorem~\ref{centralLowerBound} and Proposition \ref{Th:baLabeling} strongly suggest that, for each sufficiently large $n$, the number of  power-law graphs with $n$ vertices  is vastly larger than the number of graphs with $n$ vertices created by the  BA model.  In contrast, other generative models such as   Waxman~\cite{waxman1988routing}, N-level Hierarchical~\cite{calvert1997modeling}.
and Chung's~\cite{chung2006complex} (Chapter 3)  do not seem to have an obvious smaller label size than the one in Proposition~\ref{prop:labelingMain}.

\section{Experimental Study}

We now  perform an experimental evaluation of our labeling scheme on a number of large networks.
The source code for our experiments can be found at: \url{www.diku.dk/\~simonsen/suppmat/podc15/powerlaw.zip}

\subsection{Experimental Framework}\label{Sec:Experimental}
\paragraph{Performance Indicators.}
Recall that our labeling scheme consists of two ideas: separation of the nodes according to some threshold, and selecting a threshold depending on the power-law parameter $\alpha$.
In our labeling scheme, the threshold  is $\lceil \sqrt[\alpha]{C n/(\alpha-1)} \rceil$. We call this the \emph{predicted} threshold; it is an approximation to the theoretically optimal threshold choice when degree distributions follow the power-law curve $k\mapsto Cn/k^\alpha$ perfectly. The approximation uses integration similar to what is done in, e.g., the proof of Proposition~\ref{prop:Contained}.
For a concrete
graph $G$, it is conceivable that some other threshold $n_0$, different from the predicted threshold, would yield a labeling scheme with smaller size. 
Let $\max_t(n_0)$ and $\max_f(n_0)$ be the maximum label sizes of thin, resp.\ fat vertices in $G$ when the threshold is set at $1 \leq n_0 \leq n-1$. Clearly
the maximum label size with the threshold $n_0$ is $\max\{\max_t(n_0),\max_f(n_0)\}$. Observe further
that $\max_t(n_0)$ and $\max_f(n_0)$ are monotonically increasing, resp.\ decreasing functions of $n_0$. Hence,
the $n_0$ for which $\max\{\max_t(n_0),\max_f(n_0)\}$ is minimal is where the curves of $\max_t(n_0)$ and $\max_f(n_0)$ intersect. We call this $n_0$ the \emph{empirical} threshold.
We set up the following performance indicators to gauge (1) the difference in label size with predicted and empirical threshold, and (2) the label size obtained by our labeling scheme on several data sets, compared to other labeling schemes.

\emph{Performance Indicator 1:} We measure the label sizes for the labeling schemes with the predicted and empirical thresholds. We interpret a small relative difference between these label sizes means that the predicted threshold can achieve small label sizes without examining the global properties of the network other than the power-law parameter $\alpha$. 
 
\emph{Performance Indicator 2:} We  measure the label sizes attained by our labeling schemes to other labeling schemes, namely state-of-the art labeling schemes for the classes of bounded-degree, sparse and general graphs using the  labeling schemes suggested in \cite{adjiashvili2014labeling},  Theorem~\ref{sparse-label} and \cite{alstrup2014adjacency}. We interpret small label sizes for our scheme, especially in comparison with ``small'' classes like the class of bounded-degree graphs, as a sign that our labeling scheme efficiently utilizes the extra information about the graphs: namely that their degree distribution is reasonably well-approximated by a power-law.

\paragraph{Test Sets.}
We employ both real-world and synthetic data sets. 

The six \emph{synthetic} data sets are created by first generating a power-law degree sequence using the method of Clauset et al.~\cite[App.\ D]{clauset2009power}, subsequently constructing a corresponding graph for the sequence using the Havel-Hakimi method~\cite{hakimi1962realizability}. 
We use the range $2< \alpha < 3$ as suggested in~\cite{clauset2009power} as this range of $\alpha$ occurs most commonly in modeling of real-world networks. We generate graphs of $300,000$ and $1M.$ vertices denoted  s300$^{\alpha=x}$  and s1M$^{\alpha=x}$  respectively, for $x \in \{2.2,2.4,2.6,2.8\}$.

The three \emph{real-world} data sets originate  from articles that found the data to be well-approximated by a power-law. 
The \textsc{www} data set  \cite{albert1999internet} contains information on links between webpages within the nd.edu domain. 
The \textsc{enron} data set ~\cite{leskovec2009community}  contains email communication between  Enron employees (vertices are email addresses; there is a link between two addresses
if a mail has been sent between them).
The \textsc{internet} data set~\cite{newman} provides a snapshot the Internet structure at the level of  autonomous systems, reconstructed from BGP tables. 
For all of these sets, we consider the underlying simple, undirected graphs. For each set, standard maximum likelihood methods were used to compute the parameter
$\alpha$ of the best-fitting power-law curve \cite{clauset2009power}. Additional information on the data sets can be found in Table~\ref{t:datasets}.

\begin{table}[!ht]
\centering
\small
\begin{tabular}{lccccl}\hline
\multicolumn{6}{c}{Real-Life}\\\hline
Data set  & $\vert V \vert$ & $\vert E\vert$ & $\alpha$  & $\Delta_{\max}$ & Source\\\hline
\textsc{www}      & 325,729        & 1,117,563     & 2.16 & 10,721            & \cite{albert1999internet}\\
\textsc{enron}    &  36,692        &   183,830      & 1.97    &1,383         & \cite{leskovec2009community}\\
\textsc{internet} &  22,963        &    48,436      & 2.09     & 2,390        & \cite{newman}\\\hline

\multicolumn{6}{c}{Synthetic}\\\hline
s1M$^{\alpha=2.4}$    & 1,000,000       & 1,127,797      & 2.4    & 42,683 &-- \\
s1M$^{\alpha=2.6}$    & 1,000,000       & 878,472        & 2.6    & 12,169 &-- \\
s1M$^{\alpha=2.8}$    & 1,000,000       & 751,784         & 2.8   & 1,692  &-- \\
s300$^{\alpha=2.2}$    & 300,000        & 491,926        & 2.2    & 10,906 & --\\
s300$^{\alpha=2.4}$    & 300,000        & 327,631        & 2.4    & 3,265 & --\\
s300$^{\alpha=2.6}$    & 300,000        & 261,949        & 2.6    & 1,410 & --\\
s300$^{\alpha=2.8}$    & 300,000        & 227,247        & 2.8    & 1,842 & --\\\hline
\end{tabular}
\caption{Data sets and their properties. All graphs are undirected and simple. $\Delta_{\max}$ is the maximum degree of any vertex in the data set.}
\label{t:datasets}
\end{table}

\subsection{Findings}
Figure \ref{fig:findings} shows the distribution of maximum label sizes for one synthetic and one real-world data set. The maximum label size
for the predicted and empirical thresholds as well as upper bounds on the label sizes from different label schemes in the literature can be seen in Table \ref{t:labelsizes} for two synthetic
data sets and all three real-world data sets. 
Plots for the remaining data sets can be found in Appendix~\ref{App:ExpRes}.

\begin{figure}[!ht]
\centering
\subfloat[\small syn300$^{\alpha=2.2}$]{
    \includegraphics[width=0.5\textwidth]{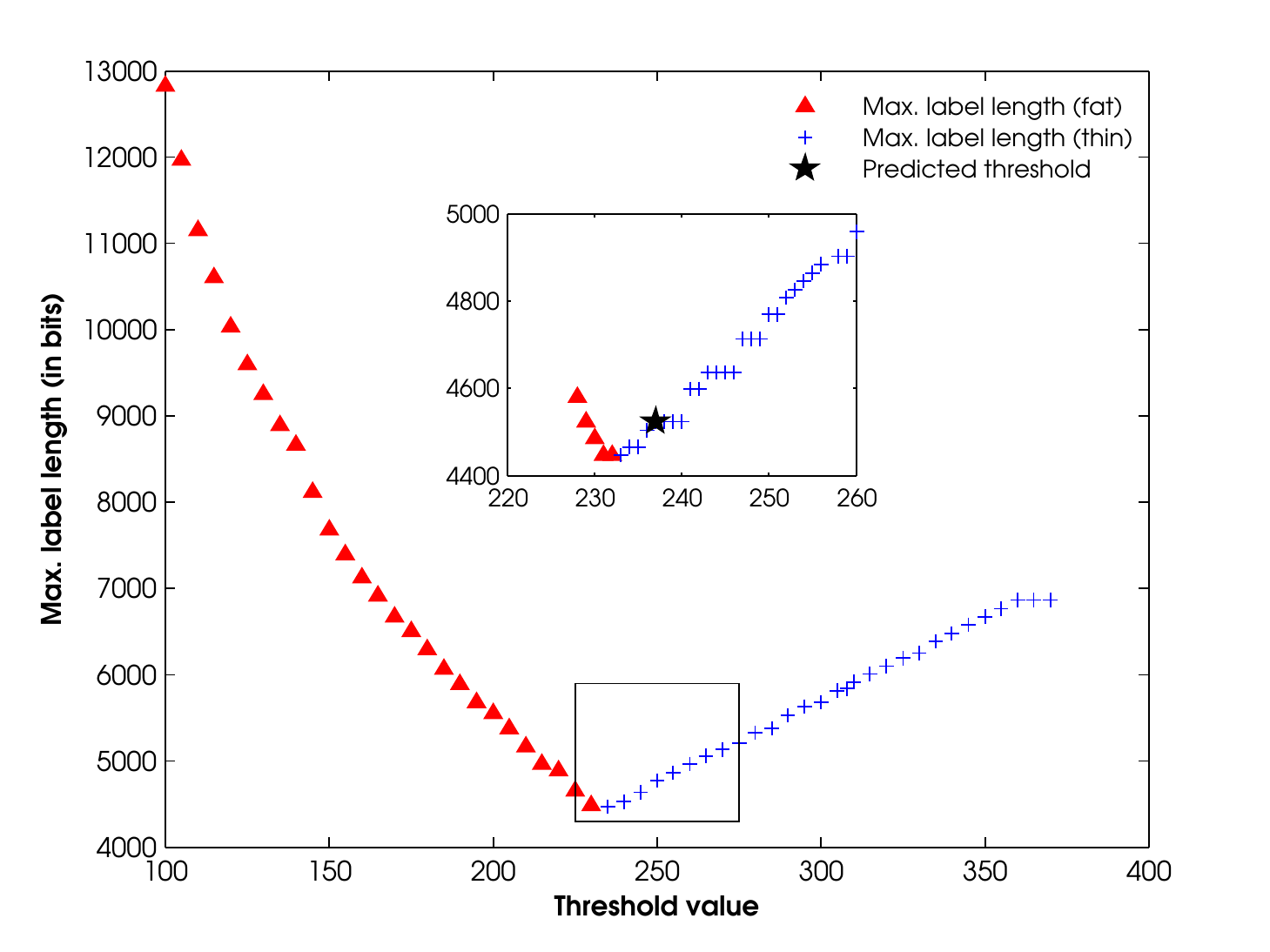}
    \label{f:fsyn300k}
}\hspace*{-2.5em}
\subfloat[\small \textsc{enron}]{
    \includegraphics[width=0.5\textwidth]{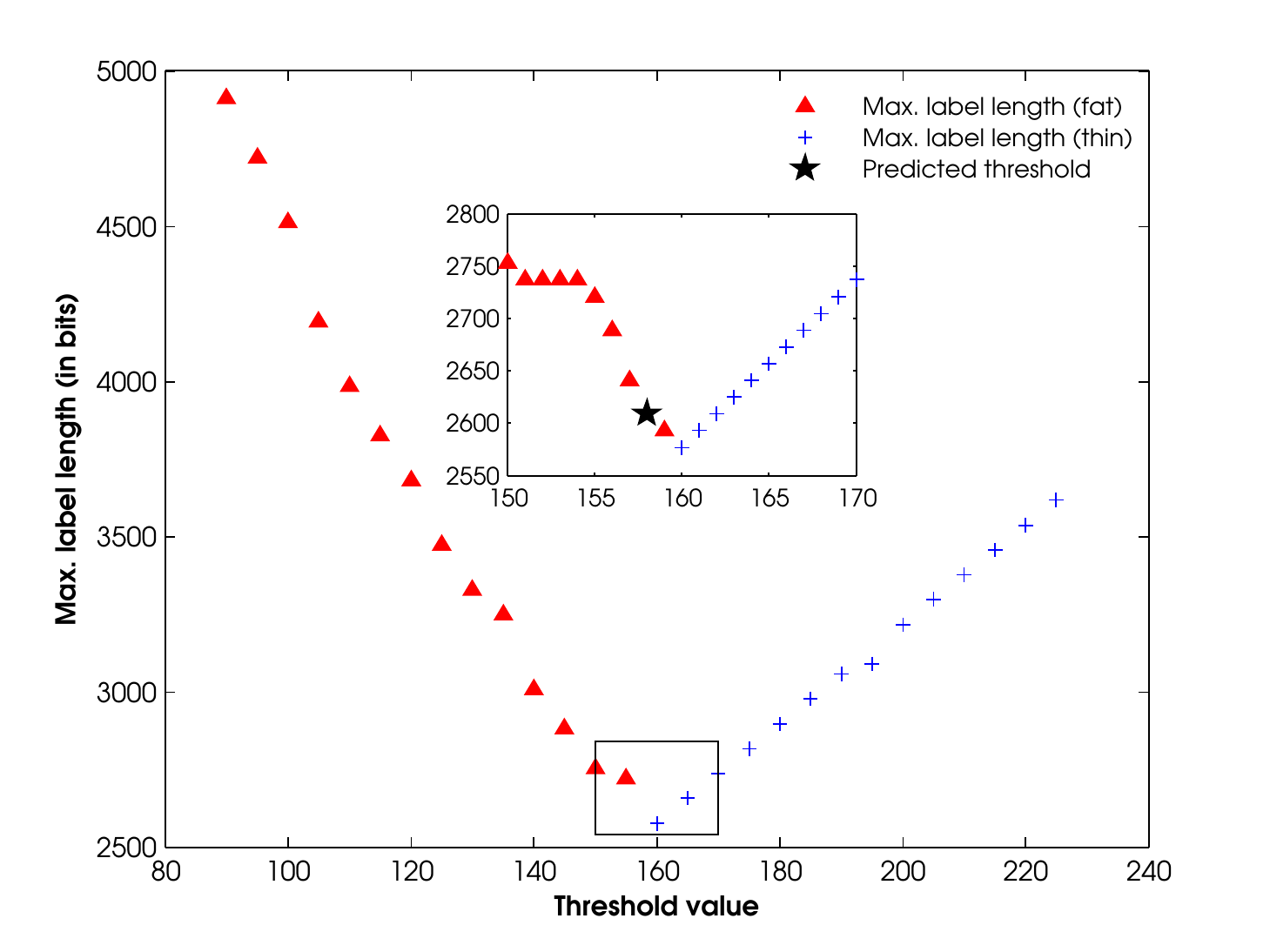}
    \label{f:fenron}
}%
\caption{Maximum label sizes of different threshold values for the   syn300$^{\alpha=2.2}$ and \textsc{enron} data sets.
The triangles and crosses represent that for the tested threshold the largest label belong to fat, resp. thin node. The star indicate the position of the predicted threshold.}
\label{fig:findings}%
\end{figure}

Table~\ref{t:labelsizes}  shows
 the maximum label sizes achieved using different labeling schemes on our data sets. ``Predicted'' shows the experimental maximum label size obtained by running our scheme on the graphs, ``Empirical'' is the label size attained by using the empirical threshold. The remaining columns show non-experimental upper bounds for different label schemes: ``Bound'' is the upper bound guaranteed in Theorem~\ref{prop:labelingMain}, ``$C$-sparse'' is  the labeling scheme for sparse graphs defined in Theorem~\ref{sparse-label}, ``BD'' is the $\lceil \frac{\Delta}{2} \rceil \lceil \log n\rceil$ bounded degree graph  labeling of~\cite{adjiashvili2014labeling}, and AKTZ is the $\lceil n/2\rceil+6$ general graph  labeling of~\cite{alstrup2014adjacency}.
Both ``Empirical'' and  ``Bound'' using simple concatenation of labels to represent the fat bit string\footnote{Our labeling schemes introduced in this paper all make use of a succinctly represented ``fat bit string''; for our experiments, we use simple concatenation of labels instead of a bit string; this incurs a $(\log n)/\alpha$ factor on the label size, but simplifies the implementation.}.

\begin{table}[!ht]
\centering
\small
\begin{tabular}{l|llllll}
Data set             &Predicted & Empirical & Bound     &  $C$-sparse & BD \cite{adjiashvili2014labeling} & AKTZ \cite{alstrup2014adjacency} \\\hline
s1M$^{\alpha=2.4}$  &$4,841$    &$4,821$    & $25,012 $ &$30,079$     &$426,820$ &$500,006$\\\hline
s1M$^{\alpha=2.6}$  &$3,361$    &$3,201$    & $15,282 $ &$26,551$     &$121,680$ &$500,006$\\\hline
s1M$^{\alpha=2.8}$  &$2,101$    &$2,061$    & $10,081 $ &$24,566$     &$16,920$  &$500,006$\\\hline
s300$^{\alpha=2.2}$ &$4,523$    &$4,447$    & $24,878 $ &$18,885$     &$103,607$ &$150,006$\\\hline
s300$^{\alpha=2.4}$ &$2,775$    &$2,680$    & $14,404 $ &$15,420$     &$31,008$  &$150,006$\\\hline
s300$^{\alpha=2.6}$ &$1,958$    &$1,920$    & $9,151 $  &$13,792$     &$13,395$  &$150,006$\\\hline
s300$^{\alpha=2.8}$ &$1,350$    &$1,312$    & $6,244 $  &$12,849$     &$17,499$  &$150,006$\\\hline
\textsc{www}        &$5,245$    &$3,060$    & $29,225 $ &$28,445$     &$101,840$ &$162,870$ \\\hline
\textsc{enron}      &$2,609$    &$2,577$    & $15,835 $ &$9,735$      &$11,056$  &$18,352$\\\hline
\textsc{internet}   &$1,426$    &$1,156$    & $8,181 $  &$4,700$      &$17,925$  &$11,487$\\\hline 
\end{tabular}
\caption{Label size in bits of labeling schemes. The two leftmost columns are experimental results; the remaining are upper bounds on label sizes computed from the characteristics of the data sets.}
\label{t:labelsizes}
\end{table}

Our findings are as follows. For Performance Indicator (i), our labeling scheme obtains maximum label size at most 3\% larger than what would have been obtained by using the empirical threshold for all synthetic data sets.
This is expected---the synthetic data sets are graphs generated specifically to have power-law distributed degree distribution. For the real-world data sets, the labeling scheme
obtains maximum label size at most 23\% larger than by using the empirical threshold; this larger deviation is likely due to degree distributions of the data sets being close to, but not quite,
power-law distributions due to natural phenomena or noise. E.g., for the \textsc{enron} data set there is sudden drop in frequency between nodes of degree $< 158$ and $\geq 158$.

For Performance Indicator (ii), both our experimental results and theoretical upper bounds for our labeling scheme are several orders of magnitudes lower than for labeling schemes aimed at more general classes of graphs, as expected. Of the more general classes of graphs, it is most interesting to compare the upper bound of bounded degree graphs---the most restrictive class of graphs that both contains the class of power-law graphs and has an efficient labeling scheme described in the literature~\cite{adjiashvili2014labeling}. As seen in Table \ref{t:labelsizes}, the upper bound on our labeling schemes for both power-law graphs and sparse graphs have better upper bounds on label sizes, but only marginally so for data sets with low maximum degree and low values of the power-law parameter $\alpha$, e.g. \textsc{Enron} ($\alpha = 1.97$). 
It is interesting to note that the actual label sizes obtained in the experiments (the two leftmost columns of Table \ref{t:labelsizes}) are substantially lower than the upper bounds, that is,
the labeling scheme performs much better in practice than suggested by theory (down to less than a kilobyte per vertex for all data sets). This phenomenon may be due to the degree distribution of the graphs of the data sets having only minor deviation from a power-law for small vertex degrees; our upper bounds on the label size are derived by using the very rich family $\PLB$ that allows very large deviation from a power-law for degrees between $1$ and $\sqrt[\alpha]{n/\log n} - 1$.

Finally, note that our labeling scheme supports adjacency for \emph{directed} graphs by using one more bit per edge in each label to store the edge orientation. For data sets whose natural interpretation is as a directed graph (e.g., the \textsc{www} set where edges are outgoing and incoming links), the results of Table \ref{t:labelsizes} thus carry over with just one more bit
added to the numbers in the two leftmost columns.

\section{Conclusion and Future Work}

We have devised adjacency labeling schemes for sparse graphs and graphs whose degree distribution approximately follows a power-law distribution. We have proven lower bounds for the class of power-law graphs showing that our labeling scheme is  almost asymptotically optimal. Furthermore, we have shown experimentally that the labeling scheme for power-law graphs obtain
results in practice requiring very little space (labels smaller than a kilobyte per vertex for real-world graphs with several hundreds of thousands of vertices).

\subsection{Future work}

It would be of interest to test the performance of the labeling scheme on more real-world data sets, and in particular investigating \emph{dynamic} labeling schemes on such sets: if vertices can enter and exit the network, labels need to be recomputed efficiently.
As our labeling scheme can be extended to handle directed graphs by using a single bit more per label, it would be interesting to investigate the overhead incurred by distributing the storage of the graph topology to the labels (as per our labeling scheme) compared to the substantial body of work on storing directed power-law graphs directly in main memory (so-called ``web-graph compression'') \cite{guillaume2002efficient,asano2003compact, asano2008efficient,claude2010fast}.
The  label sizes attained in Sec.~\ref{Sec:Experimental} can be reduced by using the succinctly represented ``fat bit string'' as well as an additional rule that prevents storing  an edge in two  labels; doing so would yield a small multiplicative reduction in label size, making our labeling scheme even more practical. 
Labeling schemes for other properties than adjacency may be investigated for power-law graphs, e.g. for distance as has been done for other classes of graphs \cite{alstrup2005labeling}
and briefly considered for power-law graphs in the context of routing algorithms \cite{chen2012compact}.
Finally, labeling schemes for power law graphs can likely be devised for the realistic case where the scheme only has incomplete knowledge of the graph, for example when the expected frequency of vertices of each degree is known, but not the exact frequency of each vertex.\newpage
 \bibliographystyle{abbrv}
\bibliography{lit}
\newpage
\appendix
\section{Experimental results in detail}\label{App:ExpRes}
Subsections ~\ref{App:ExpRes:MaxLabelSyn} and \ref{App:ExpRes:MaxLabelEmp}  show the maximum label sizes for all synthetic and real-world data sets, respectively.

\subsection{Maximum label size distribution for synthetic datasets}\label{App:ExpRes:MaxLabelSyn}
\begin{figure}[!ht]
\centering
\subfloat[\small syn300$^{\alpha=2.2}$]{
    \includegraphics[width=0.5\textwidth]{synthetic-300k-alpha22.pdf}
}
\subfloat[\small syn300$^{\alpha=2.4}$]{
    \includegraphics[width=0.5\textwidth]{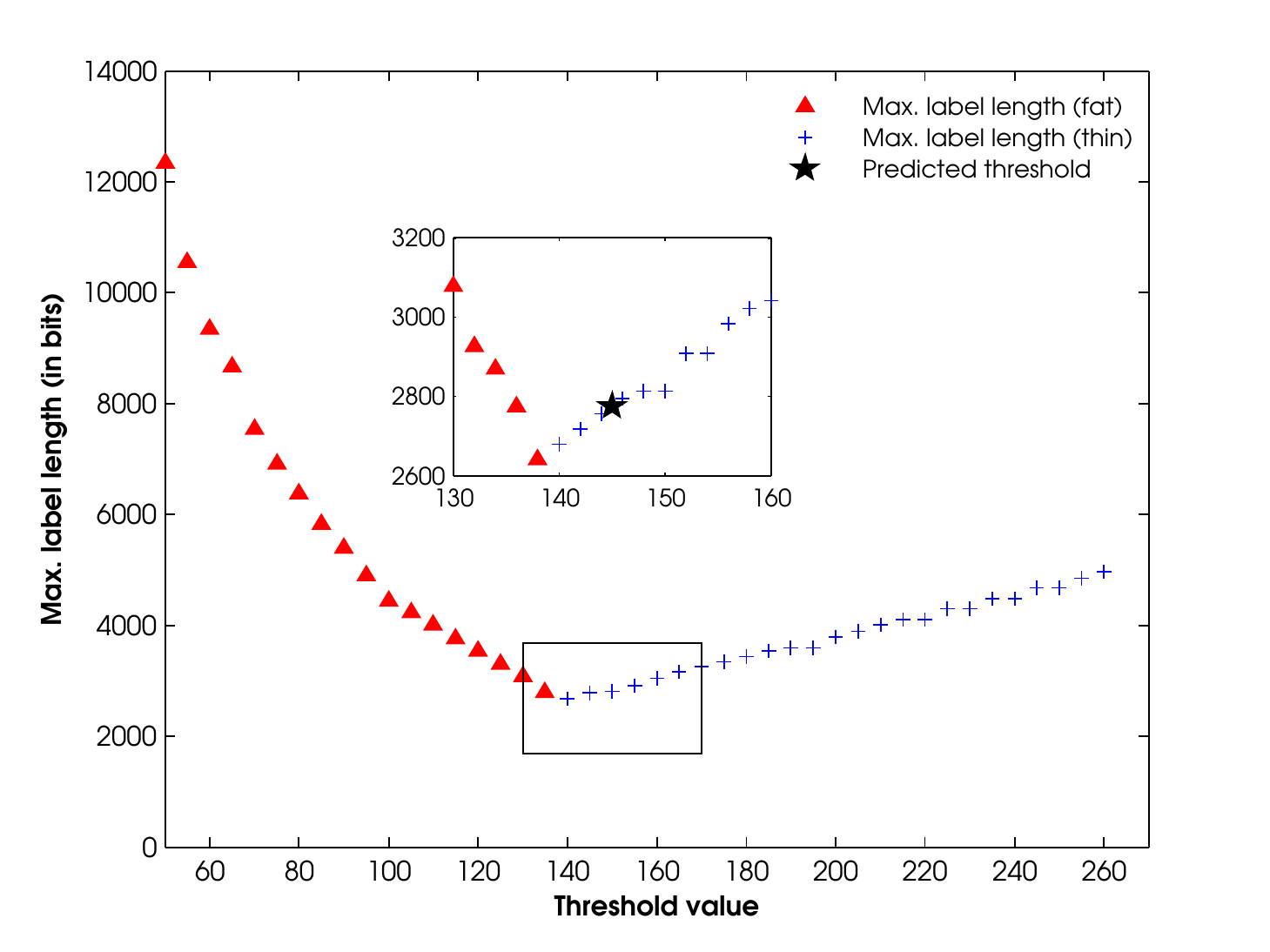}
}%
\quad
\subfloat[\small syn300$^{\alpha=2.6}$]{
    \includegraphics[width=0.5\textwidth]{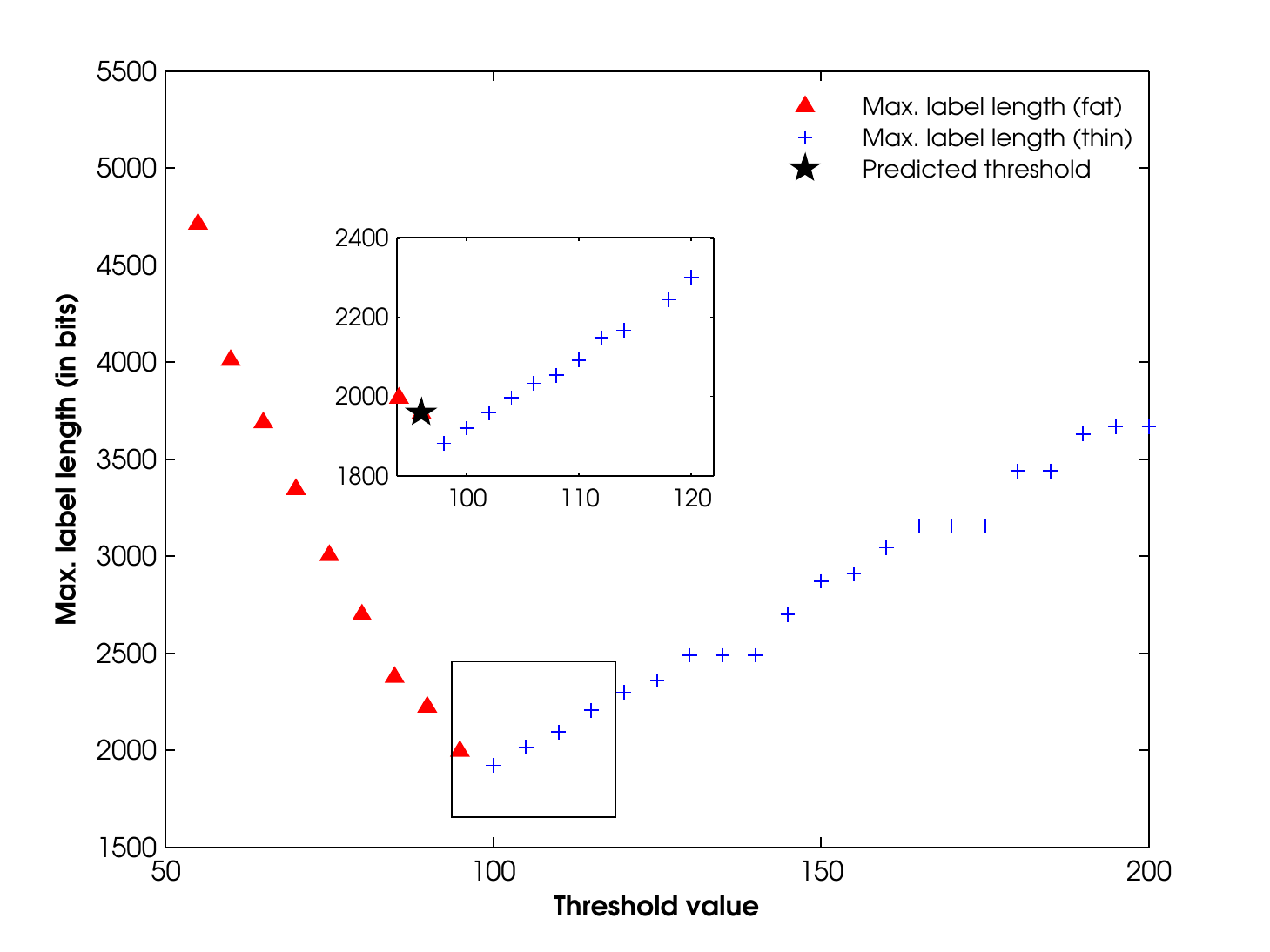}
}
\subfloat[\small syn300$^{\alpha=2.8}$]{
    \includegraphics[width=0.5\textwidth]{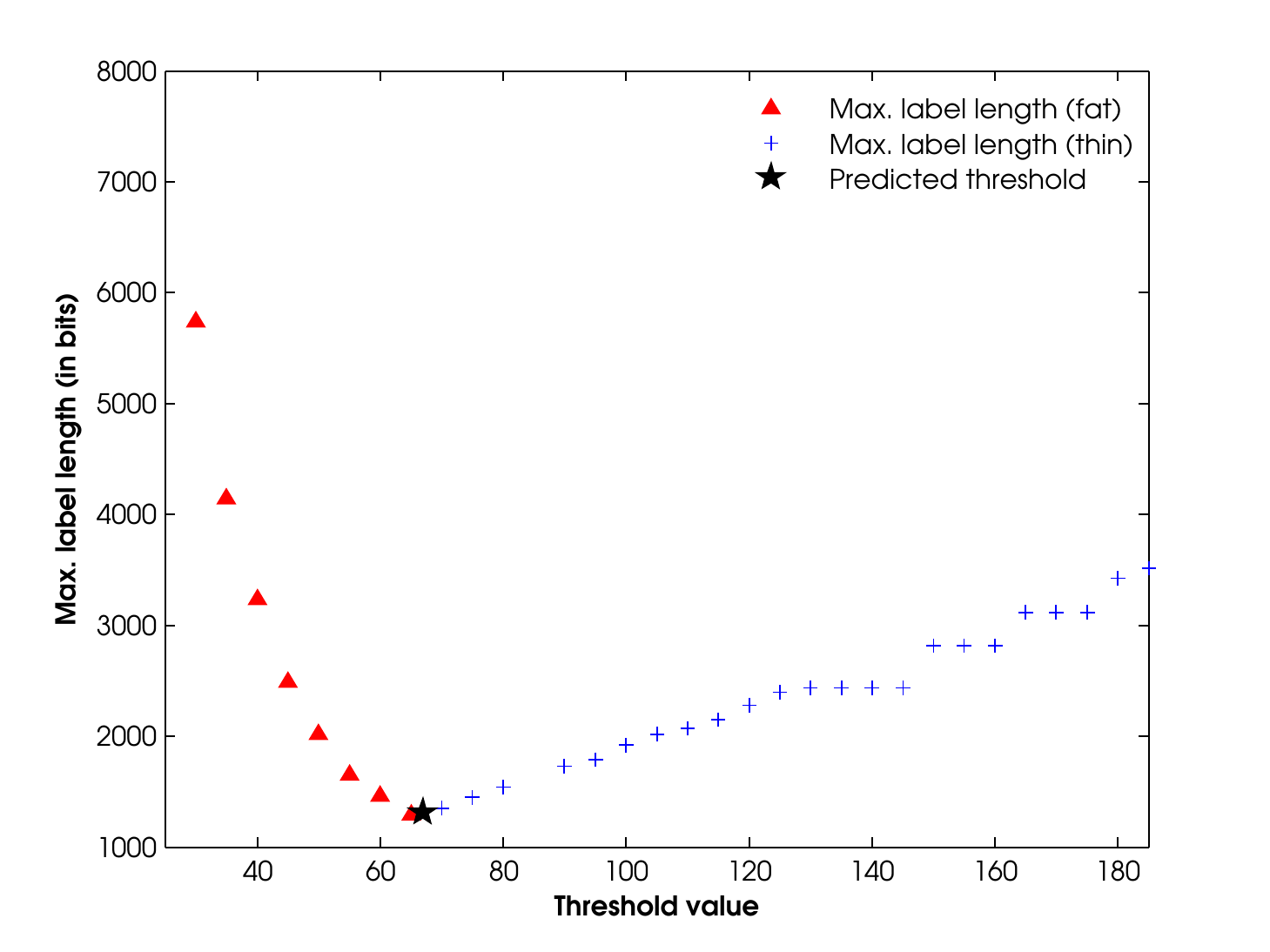}
}%
\caption{Distribution of maximum label sizes for four different synthetic datasets of $\vert V\vert = 300,000$. Each dataset was generated using one of $\alpha$-values: $2.2,~2.4,~2.6,~2.8$. Fat vertices are shown as red triangles and thin vertices as blue crosses. The black pentagram shows the label size obtained by using the \emph{predicted} threshold. The transition between fat and thin vertices is the maximum label size obtained by using the empirical threshold.}%
\label{fig:synthetic300}%
\end{figure}

\begin{figure}[!ht]
\centering
\subfloat[\small syn1M$^{\alpha=2.4}$]{
    \includegraphics[width=0.5\textwidth]{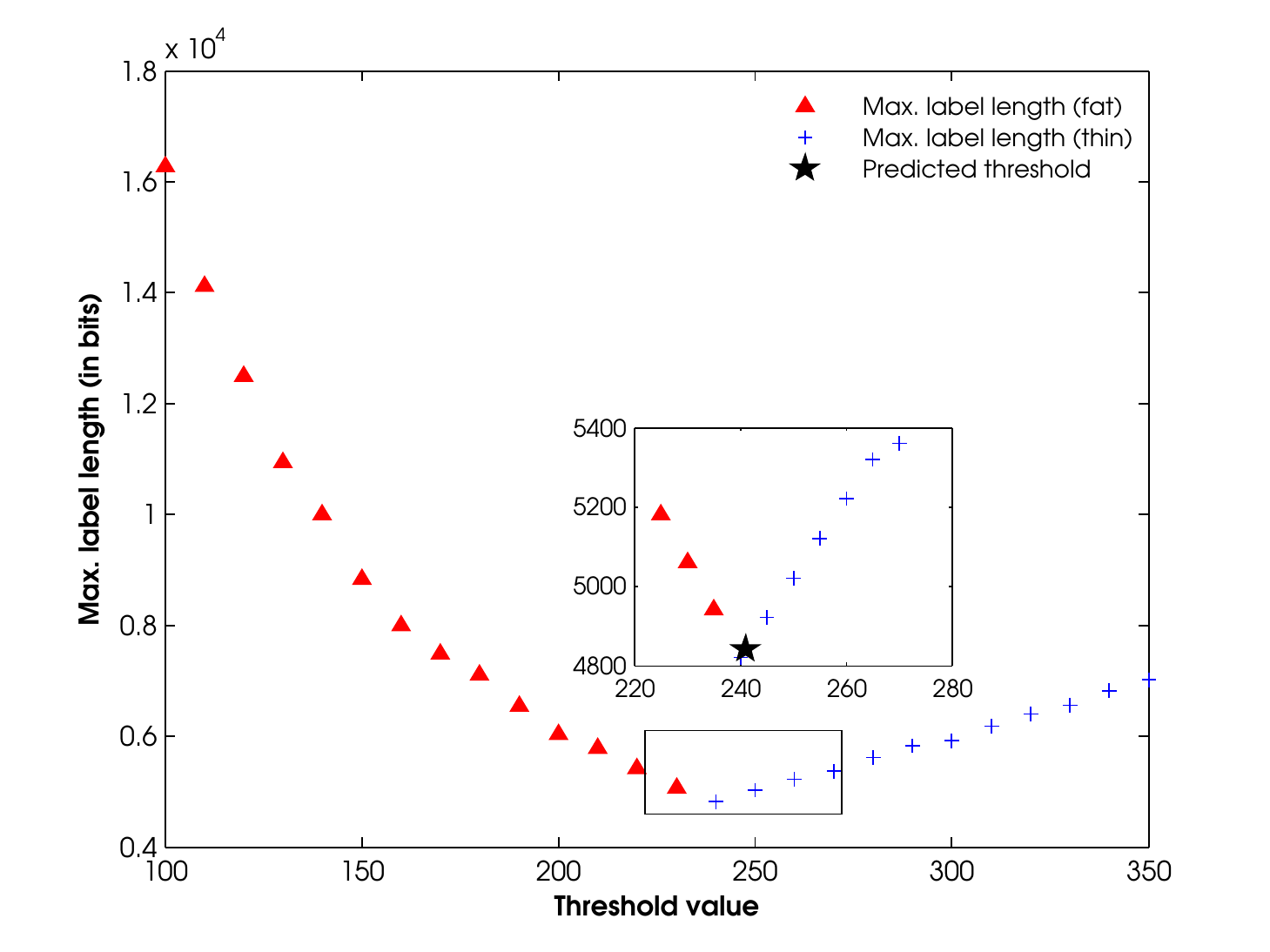}
}%
\subfloat[\small syn1M$^{\alpha=2.6}$]{
    \includegraphics[width=0.5\textwidth]{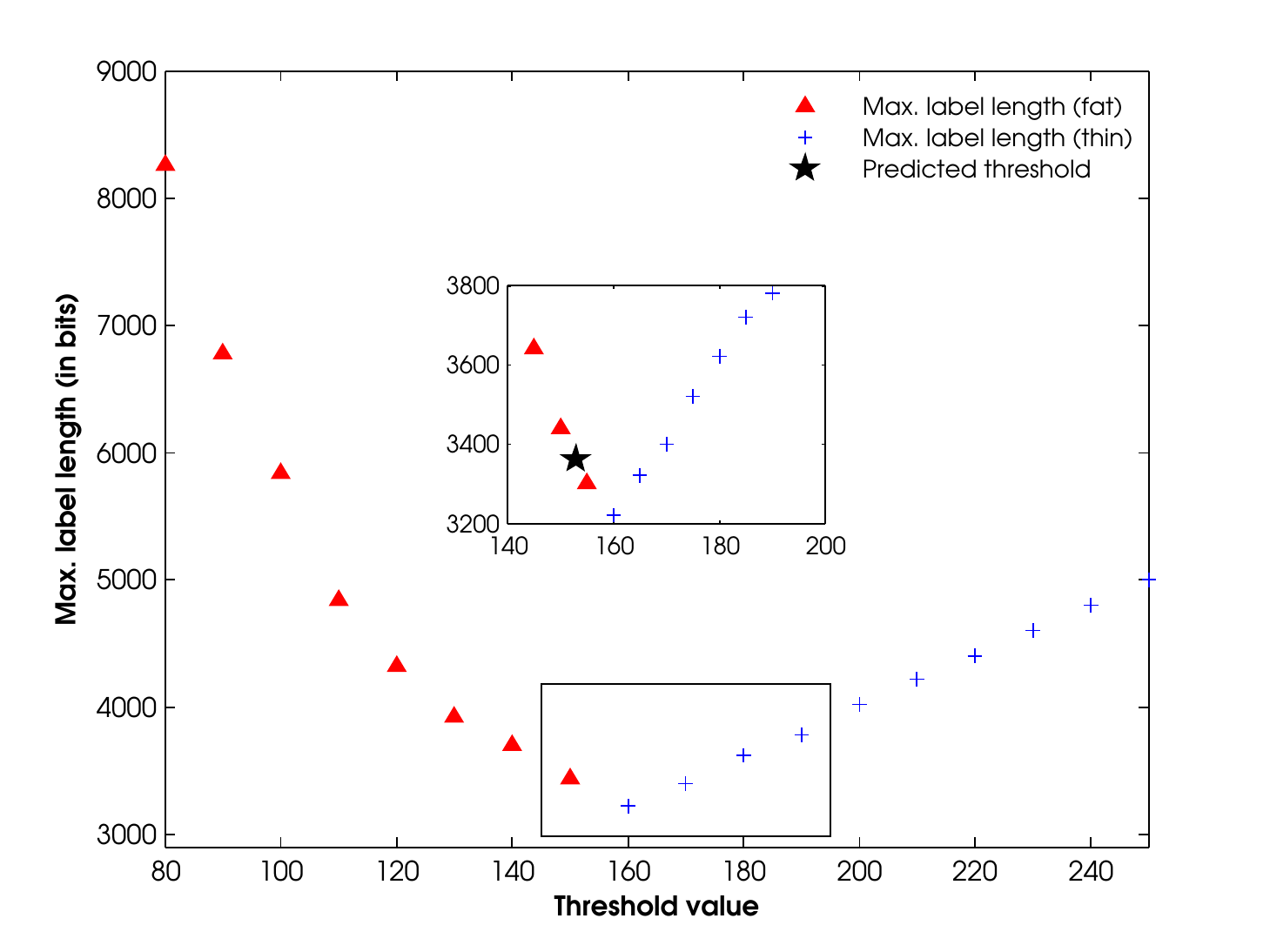}
}
\quad
\subfloat[\small syn1M$^{\alpha=2.8}$]{
    \includegraphics[width=0.5\textwidth]{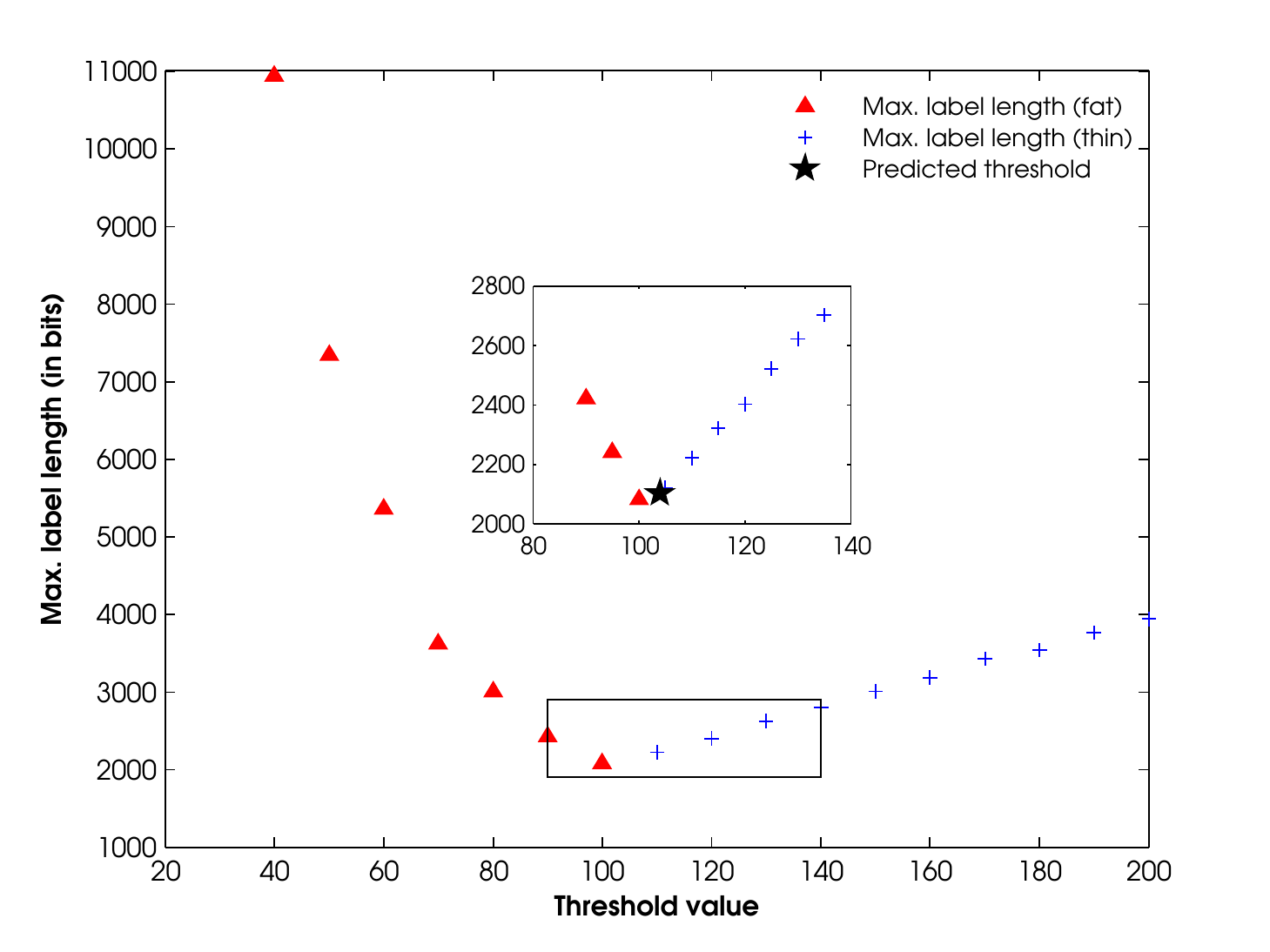}
}%
\caption{Distribution of maximum label sizes for three different synthetic datasets of $\vert V\vert = 1,000,000$. Each dataset was generated using one of $\alpha$-values: $2.4,~2.6,~2.8$. Fat vertices are shown as red triangles and thin vertices as blue crosses. The black pentagram shows the label size obtained by using the \emph{predicted} threshold. The transition between fat and thin vertices is the maximum label size obtained by using the empirical threshold.}%
\label{fig:synthetic1M}%
\end{figure}
\clearpage

\subsection{Maximum label size distribution for real-life datasets}\label{App:ExpRes:MaxLabelEmp}
For completeness,  we  provide an illustration of the best-fitting power law fitted to the probability mass function of the data.

\begin{figure}[!ht]
\centering
\subfloat[\small Fat and thin vertices vs. threshold values]{
    \includegraphics[width=0.5\textwidth]{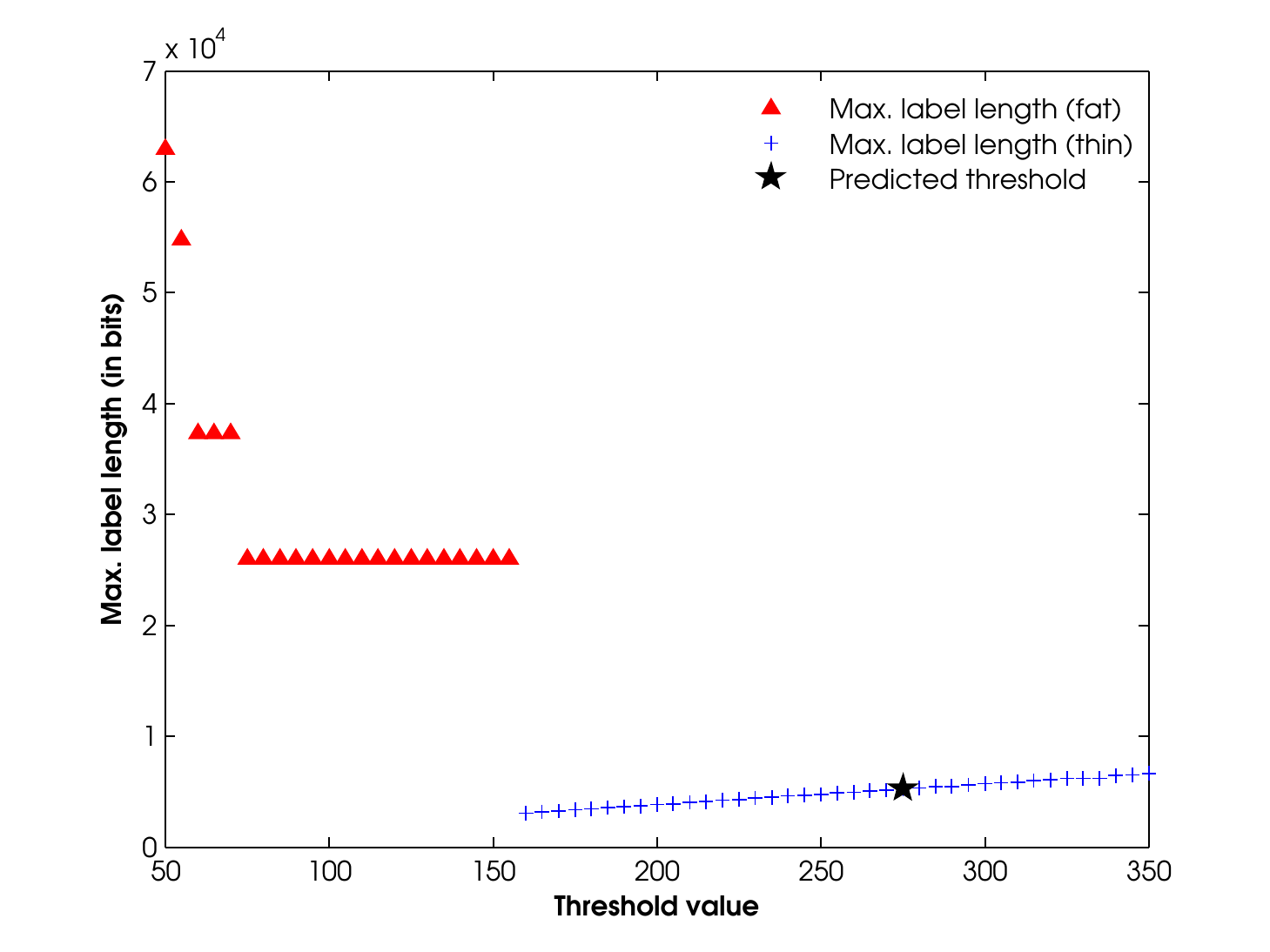}
}%
\subfloat[\small Power law fit]{
    \includegraphics[width=0.5\textwidth]{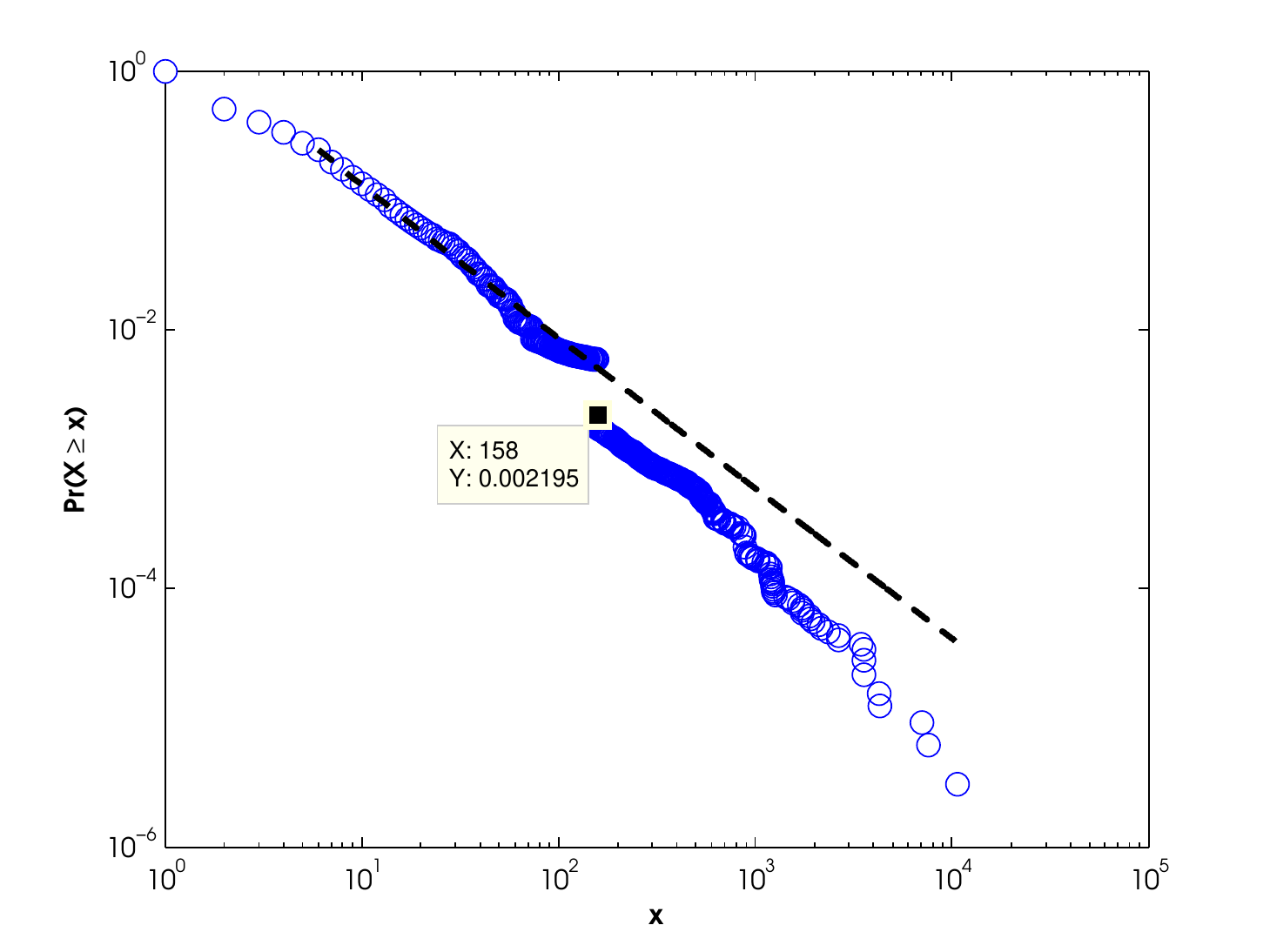}
}%
\caption{Left: Fat and thin vertices plotted against increasing threshold values for the \textsc{www} dataset. The black pentagram shows the predicted threshold ($1/\zeta(\alpha)\sqrt[\alpha](n)$) rounded to nearest integer. Right: Best-fitting power law ($\alpha = 2.16$) superimposed on the complementary cumulative distribution function (CCDF) using the framework by \cite{clauset2009power}.} %
\label{fig:www}%
\end{figure}

\begin{figure}[!ht]
\centering
\subfloat[\small Fat and thin vertices vs. threshold values]{
    \includegraphics[width=0.5\textwidth]{enron-mail.pdf}
}%
\subfloat[\small Power law fit]{
    \includegraphics[width=0.5\textwidth]{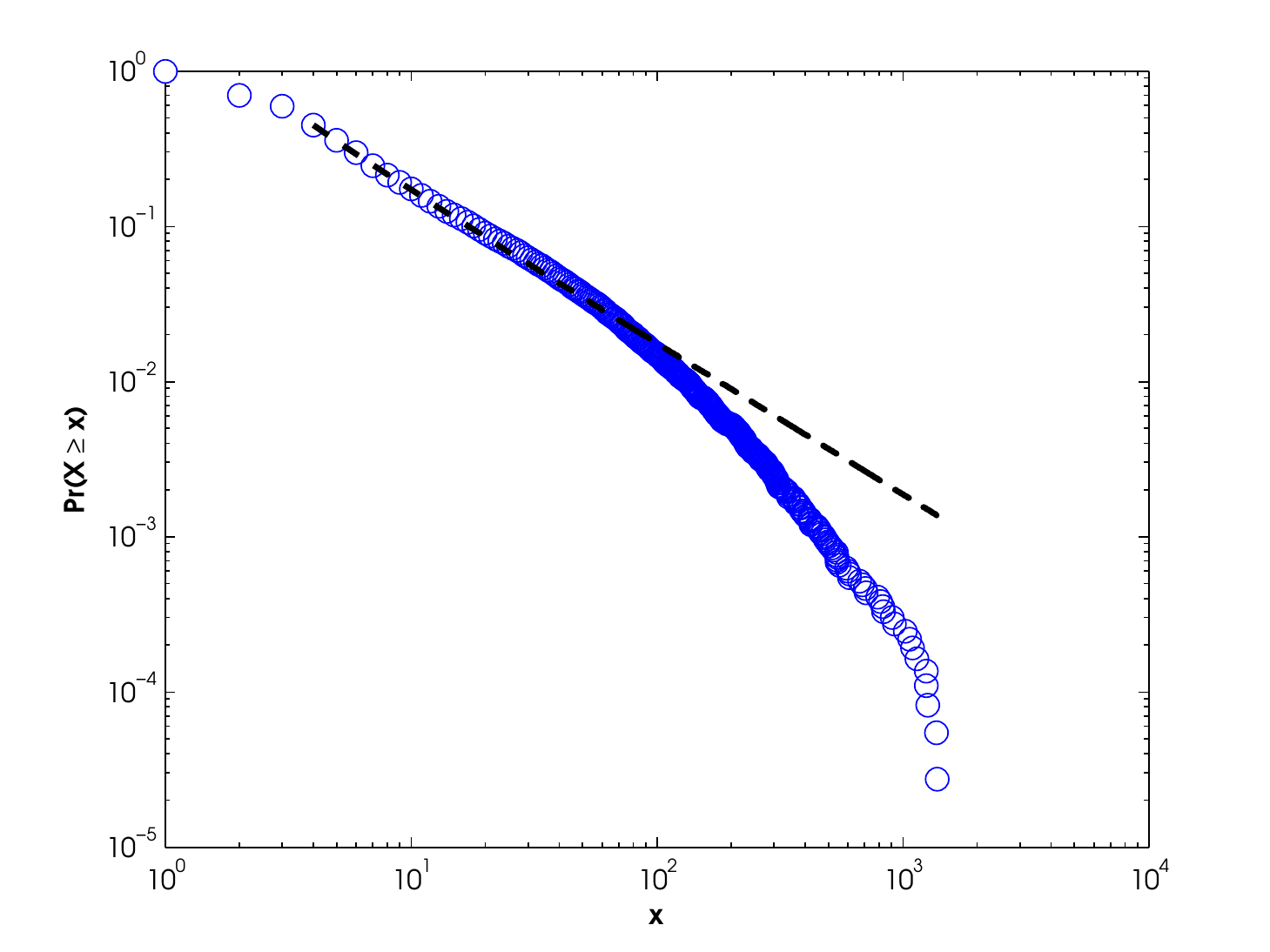}
}%
\caption{Left: Fat and thin vertices plotted against increasing threshold values for the \textsc{enron} email communication dataset. The black pentagram is the predicted threshold ($1/\zeta(\alpha)\sqrt[\alpha](n)$) rounded to the nearest integer. Right: Right: Best-fitting power law ($\alpha = 1.97$)  superimposed on the complementary cumulative distribution function (CCDF) using the framework by \cite{clauset2009power}.}
\label{fig:enron}%
\end{figure}

\begin{figure}[!ht]
\centering
\subfloat[\small Fat and thin vertices vs. threshold values]{
    \includegraphics[width=0.5\textwidth]{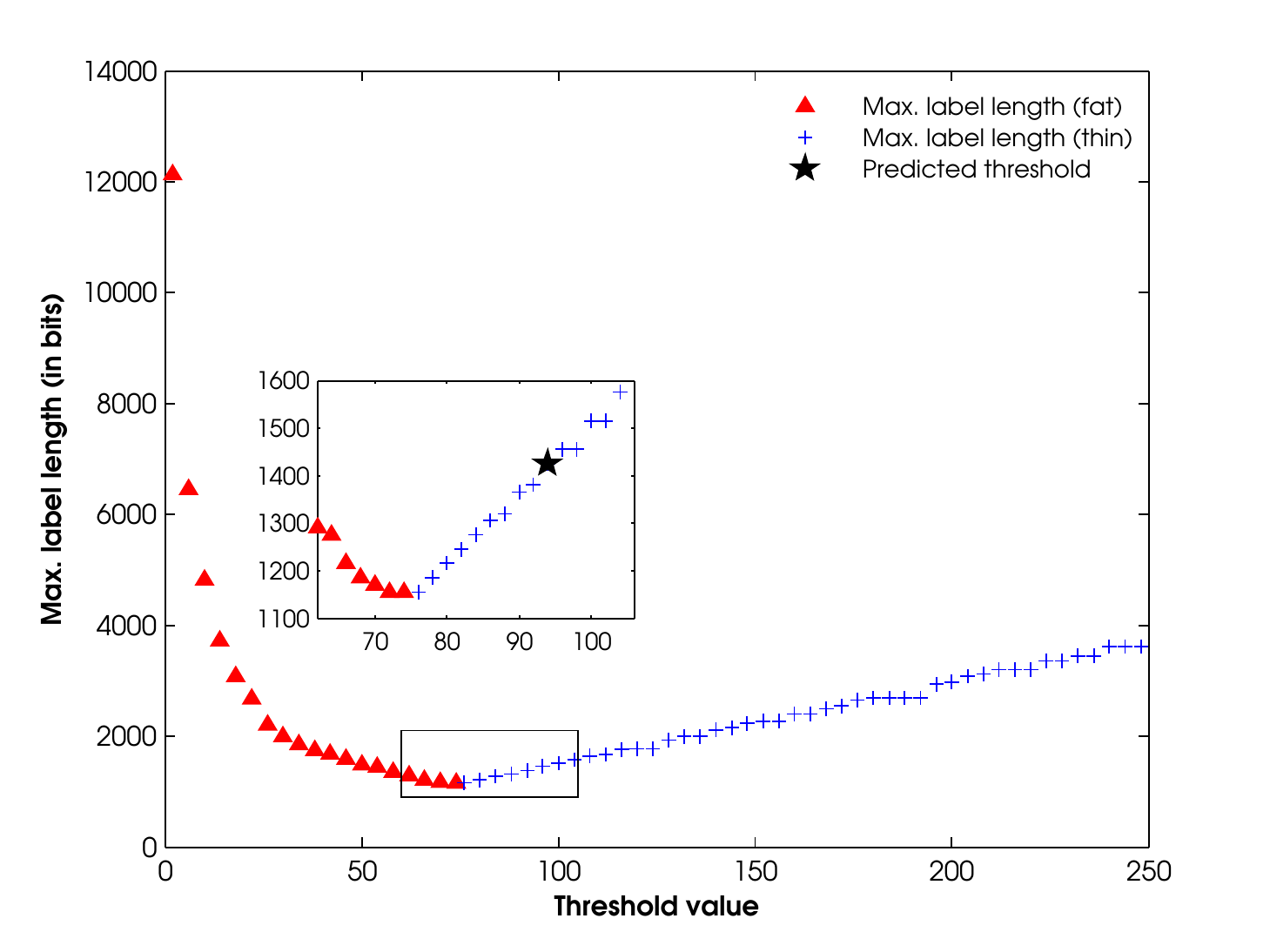}
}%
\subfloat[\small Power law fit]{
    \includegraphics[width=0.5\textwidth]{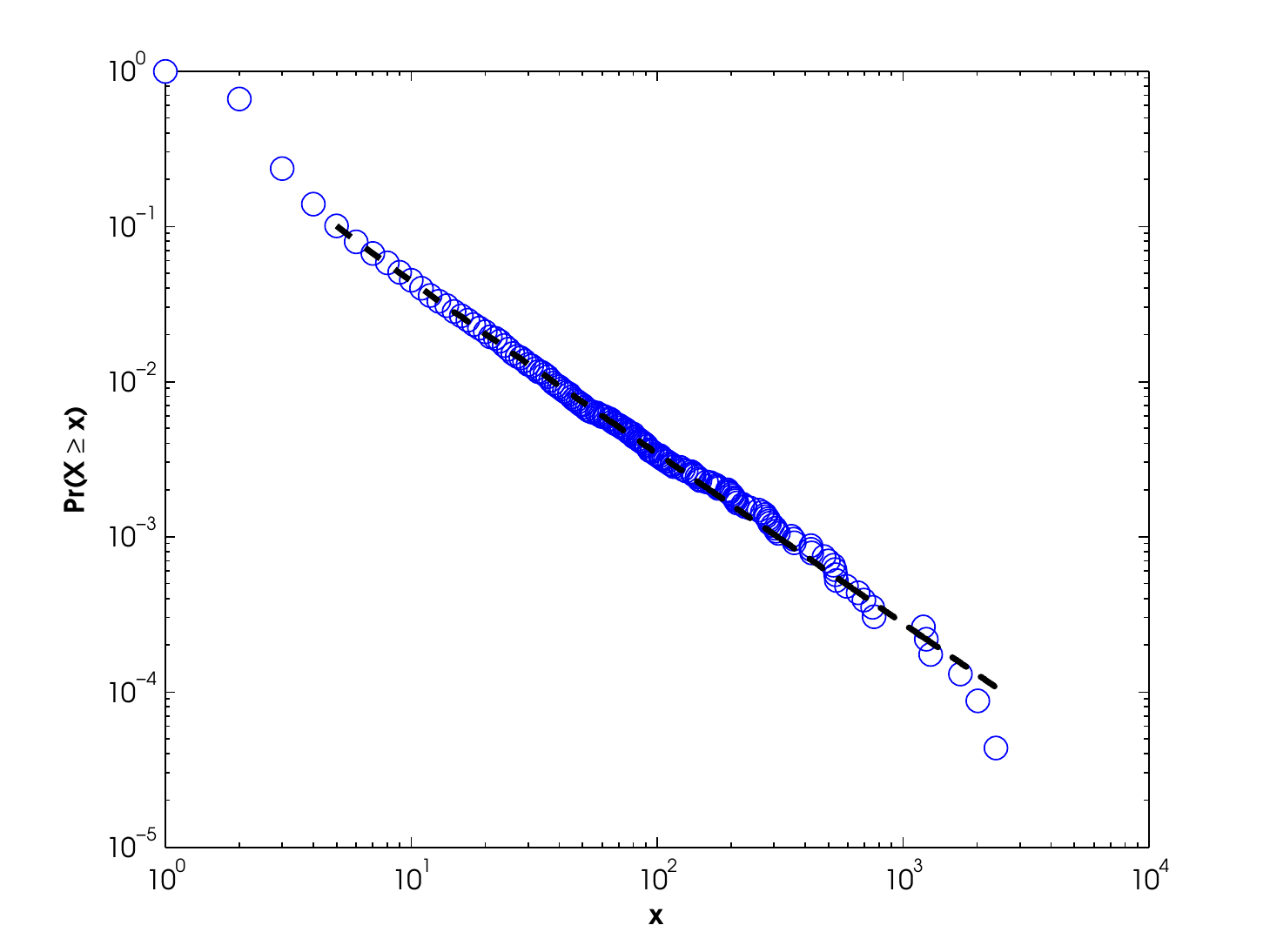}
}%
\caption{Left: Fat and thin vertices plotted against increasing threshold values for the \textsc{internet} dataset. The black pentagram is the predicted threshold ($1/\zeta(\alpha)\sqrt[\alpha](n)$) rounded to nearest integer. Right: Right: Best-fitting power law ($\alpha = 2.09$) superimposed on the complementary cumulative distribution function (CCDF) using the framework by \cite{clauset2009power}.} %
\label{fig:internet}%
\end{figure}

\end{document}